\documentclass[10pt,conference]{IEEEtran}
\IEEEoverridecommandlockouts
\usepackage{cite}
\usepackage{amsmath,amssymb,amsfonts}
\usepackage{algorithmic}
\usepackage{graphicx}
\usepackage{textcomp}
\usepackage{xcolor}

\usepackage[nodisplayskipstretch]{setspace}
\usepackage[T1]{fontenc}
\usepackage[utf8]{inputenc}
\usepackage[ngerman,english]{babel}
\usepackage{amsmath, amssymb, ﻿amsfonts}   
\usepackage{cuted} 
\usepackage{graphicx}         
\usepackage{algorithmic}
\usepackage{algorithm}
\usepackage{stfloats}

\usepackage[hyphens]{url}
\usepackage[hidelinks]{hyperref}
\hypersetup{colorlinks=false,breaklinks=true}

\usepackage{verbatim}
\usepackage{graphicx}
\usepackage{cite}
\usepackage[top=0.7in, left=0.67in, right=0.67in, bottom=1in]{geometry}

\usepackage{upgreek}			  
\usepackage{mathtools}		  
\usepackage{url}					
\usepackage{enumitem}		
\usepackage{accents}			
\usepackage{bbm}
\usepackage{xcolor}
\usepackage{lipsum}
\usepackage{dirtytalk}
\usepackage{xcolor}
\usepackage{float}
\usepackage{accents}
\usepackage[export]{adjustbox}
\usepackage{setspace}

\usepackage{subfig}

\usepackage{pifont} 

\usepackage{booktabs, tabularx}
\usepackage{array}

\usepackage{amsthm}
\usepackage{enumitem}

\usepackage{tikz}
\usetikzlibrary{decorations.markings,chains,matrix}
\usetikzlibrary{intersections,calc}
\usepackage{tikz-3dplot} 
\usepackage{pgfplots}
\usepackage{pgfplotstable}
\pgfplotsset{compat = newest}
\usepackage{datatool}
\usepackage{textcomp}
\usetikzlibrary{shapes,arrows}
\usetikzlibrary{decorations.pathreplacing,angles,quotes}

\usepackage[pscoord]{eso-pic} 
\newtheorem{thm}{Theorem}

\newtheorem{cor}{Corollary}
\newtheorem{defn}{Definition}

\newtheorem{rmk}{Remark}

\DeclareMathOperator{\FDR}{FDR}

\DeclareMathOperator{\TPR}{TPR}

\DeclareMathOperator{\Var}{Var}
\DeclareMathOperator{\SNR}{SNR}

\DeclareMathOperator{\EN}{EN}
\DeclareMathOperator{\IEN}{IEN}
\DeclareMathOperator{\sign}{sign}

\newtheorem{innercustomgeneric}{\customgenericname}
\providecommand{\customgenericname}{}
\newcommand{\newcustomtheorem}[2]{%
  \newenvironment{#1}[1]
  {%
   \renewcommand\customgenericname{#2}%
   \renewcommand\theinnercustomgeneric{##1}%
   \innercustomgeneric
  }
  {\endinnercustomgeneric}
}

\newcustomtheorem{customthm}{Theorem}
\newcustomtheorem{customlemma}{Lemma}
\newcustomtheorem{customcor}{Corollary}
\newcustomtheorem{customassumption}{Assumption}

\newcommand{\y}{\boldsymbol{y}}
\newcommand{\x}{\boldsymbol{x}}

\newcommand{\X}{\boldsymbol{X}}

\newcommand{\bbeta}{\boldsymbol{\beta}}
\newcommand{\bepsilon}{\boldsymbol{\epsilon}}
\newcommand{\hatbbeta}{\boldsymbol{\hat{\beta}}}

\newcommand{\A}{\mathcal{A}}

\newcommand{\XK}{\boldsymbol{\protect\accentset{\circ}{X}}}

\newcommand{\XWK}{\boldsymbol{\widetilde{X}}}

\newcommand{\C}{\mathcal{C}}

\newcommand{\G}{\mathcal{G}}
\newcommand{\1}{\boldsymbol{1}}

\makeatletter
\let\save@mathaccent\mathaccent
\newcommand*\if@single[3]{%
  \setbox0\hbox{${\mathaccent"0362{#1}}^H$}%
  \setbox2\hbox{${\mathaccent"0362{\kern0pt#1}}^H$}%
  \ifdim\ht0=\ht2 #3\else #2\fi
  }
\newcommand*\rel@kern[1]{\kern#1\dimexpr\macc@kerna}
\newcommand*\widebar[1]{\@ifnextchar^{{\wide@bar{#1}{0}}}{\wide@bar{#1}{1}}}
\newcommand*\wide@bar[2]{\if@single{#1}{\wide@bar@{#1}{#2}{1}}{\wide@bar@{#1}{#2}{2}}}
\newcommand*\wide@bar@[3]{%
  \begingroup
  \def\mathaccent##1##2{%
    \let\mathaccent\save@mathaccent
    \if#32 \let\macc@nucleus\first@char \fi
    \setbox\z@\hbox{$\macc@style{\macc@nucleus}_{}$}%
    \setbox\tw@\hbox{$\macc@style{\macc@nucleus}{}_{}$}%
    \dimen@\wd\tw@
    \advance\dimen@-\wd\z@
    \divide\dimen@ 3
    \@tempdima\wd\tw@
    \advance\@tempdima-\scriptspace
    \divide\@tempdima 10
    \advance\dimen@-\@tempdima
    \ifdim\dimen@>\z@ \dimen@0pt\fi
    \rel@kern{0.6}\kern-\dimen@
    \if#31
      \overline{\rel@kern{-0.6}\kern\dimen@\macc@nucleus\rel@kern{0.4}\kern\dimen@}%
      \advance\dimen@0.4\dimexpr\macc@kerna
      \let\final@kern#2%
      \ifdim\dimen@<\z@ \let\final@kern1\fi
      \if\final@kern1 \kern-\dimen@\fi
    \else
      \overline{\rel@kern{-0.6}\kern\dimen@#1}%
    \fi
  }%
  \macc@depth\@ne
  \let\math@bgroup\@empty \let\math@egroup\macc@set@skewchar
  \mathsurround\z@ \frozen@everymath{\mathgroup\macc@group\relax}%
  \macc@set@skewchar\relax
  \let\mathaccentV\macc@nested@a
  \if#31
    \macc@nested@a\relax111{#1}%
  \else
    \def\gobble@till@marker##1\endmarker{}%
    \futurelet\first@char\gobble@till@marker#1\endmarker
    \ifcat\noexpand\first@char A\else
      \def\first@char{}%
    \fi
    \macc@nested@a\relax111{\first@char}%
  \fi
  \endgroup
}
\makeatother

\DeclareFontFamily{U}{dutchcal}{\skewchar\font=45}
\DeclareFontShape{U}{dutchcal}{m}{n}{<-> s*[1.2] dutchcal-r}{}
\DeclareFontShape{U}{dutchcal}{b}{n}{<-> s*[1.2] dutchcal-b}{}
\DeclareMathAlphabet{\mathdutchcal}{U}{dutchcal}{m}{n}
\SetMathAlphabet{\mathdutchcal}{bold}{U}{dutchcal}{b}{n}
\DeclareMathAlphabet{\mathdutchbcal}{U}{dutchcal}{b}{n}

\newlist{steps}{enumerate}{1}
\setlist[steps, 1]{label = {Step \arabic*:}, ref = {Step \arabic*}}

\newlist{alglist}{enumerate}{1}
\setlist[alglist, 1]{label = {\arabic*.}, ref = {\arabic*}}

\newcolumntype{L}[1]{>{\raggedright\arraybackslash}p{#1}}
\newcolumntype{C}[1]{>{\centering\arraybackslash}p{#1}}
\newcolumntype{R}[1]{>{\raggedleft\arraybackslash}p{#1}}

\definecolor{dark_green}{RGB}{102,166,30}
\definecolor{applegreen}{rgb}{0.55, 0.71, 0.0}

\definecolor{dark_red}{RGB}{217,95,2}
\definecolor{bittersweet}{rgb}{1.0, 0.44, 0.37}

\definecolor{dark_yellow}{RGB}{230,171,2}
\definecolor{bananayellow}{rgb}{1.0, 0.88, 0.21}

\newcommand{\placetextbox}[3]{
  \setbox0=\hbox{#3}
  \AddToShipoutPictureFG*{
    \put(\LenToUnit{#1\paperwidth},\LenToUnit{#2\paperheight}){\vtop{{\null}\makebox[0pt][c]{#3}}}%
  }%
}%

\mathtoolsset{showonlyrefs=true} 
\restylefloat{table}

\def\BibTeX{{\rm B\kern-.05em{\sc i\kern-.025em b}\kern-.08em
    T\kern-.1667em\lower.7ex\hbox{E}\kern-.125emX}}
\begin{document}

\title{The Informed Elastic Net for Fast Grouped Variable Selection and FDR Control in Genomics Research}

\author{

\IEEEauthorblockN{Jasin Machkour}
\IEEEauthorblockA{
\textit{Technische Universit\"at Darmstadt} \\
64283 Darmstadt, Germany \\
jasin.machkour@tu-darmstadt.de}
\and
\IEEEauthorblockN{Michael Muma}
\IEEEauthorblockA{
\textit{Technische Universit\"at Darmstadt} \\
64283 Darmstadt, Germany \\
michael.muma@tu-darmstadt.de}
\and
\IEEEauthorblockN{Daniel P. Palomar}
\IEEEauthorblockA{
\scalebox{0.99}[1]{\textit{The Hong Kong University of Science and Technology}}\\
Clear Water Bay, Hong Kong SAR, China\\
palomar@ust.hk}

\thanks{The first and second author are supported by the LOEWE initiative (Hesse, Germany) within the emergenCITY center. The second author is also supported by the ERC Starting Grant ScReeningData. The third author is supported by the Hong Kong GRF 16207820 research grant.}
\thanks{Extensive calculations on the Lichtenberg High-Performance Computer of the Technische Universität Darmstadt were conducted for this research.}
}

\maketitle

\begin{abstract}
Modern genomics research relies on genome-wide association studies (GWAS) to identify the few genetic variants among potentially millions that are associated with diseases of interest. Only reproducible discoveries of groups of associations improve our understanding of complex polygenic diseases and enable the development of new drugs and personalized medicine. Thus, fast multivariate variable selection methods that have a high true positive rate (TPR) while controlling the false discovery rate (FDR) are crucial. Recently, the T-Rex+GVS selector, a version of the T-Rex selector that uses the elastic net (EN) as a base selector to perform grouped variable election, was proposed. Although it significantly increased the TPR in simulated GWAS compared to the original T-Rex, its comparably high computational cost limits scalability. Therefore, we propose the informed elastic net (IEN), a new base selector that significantly reduces computation time while retaining the grouped variable selection property. We quantify its grouping effect and derive its formulation as a Lasso-type optimization problem, which is solved efficiently within the T-Rex framework by the terminated LARS algorithm. Numerical simulations and a GWAS study demonstrate that the proposed T-Rex+GVS (IEN) exhibits the desired grouping effect, reduces computation time, and achieves the same TPR as T-Rex+GVS (EN) but with lower FDR, which makes it a promising method for large-scale GWAS.
\end{abstract}

\placetextbox{0.5}{0.08}{\fbox{\parbox{\dimexpr\textwidth-2\fboxsep-2\fboxrule\relax}{\footnotesize Published in IEEE International Workshop on Computational Advances in Multi-Sensor Adaptive Processing (CAMSAP), 10-13 December 2023, Los Sue\~nos, Costa Rica.}}}

\section{Introduction}
\label{sec: Introduction}
Fast grouped variable selection is essential in many modern signal processing applications such as genomics research~\cite{gwasCatalog}, direction-of-arrival (DOA) estimation~\cite{tan2014direction}, and financial index tracking~\cite{benidis2017sparse}, where groups of highly correlated variables are present in the data~\cite{reich2001linkage,shan1985spatial,mantegna1999introduction}. In this paper, we focus on modern genomics research that aims at understanding complex genetic diseases in order to enable personalized medicine and the development of new drugs. In genomics research, a useful and widespread tool for the discovery of the few single nucleotide polymorphisms (SNPs) (among potentially millions of SNPs) that are associated with a disease of interest are genome-wide association studies (GWAS)~\cite{gwasCatalog}. More specifically, in GWAS the aim is to detect as many as possible of the SNPs that are associated with a disease of interest while keeping the number of false discoveries low. Thus, false discovery rate (FDR)-controlling variable selection methods that exhibit a high true positive rate (TPR) are required. The FDR is defined as the expected fraction of false discoveries among all discoveries, i.e., $\FDR = \mathbb{E}[ \text{\# False discoveries } / \text{ \# Discoveries}]$, while the TPR is defined as the expected fraction of true discoveries among all true active variables, i.e., $\TPR = \mathbb{E}[ \text{\# True discoveries } / \text{ \# True actives}]$. However, since the number of SNPs $p$ (i.e., variables) usually exceeds the number of subjects $n$ (i.e., samples), the existing FDR-controlling methods for the low-dimensional setting~\cite{benjamini1995controlling,benjamini2001control,barber2015controlling,dai2016knockoff} are not applicable in this high-dimensional setting (i.e., $p > n$). In recent years, two multivariate FDR-controlling variable selection frameworks for the high-dimensional setting have been proposed: the \textit{T-Rex} selector~\cite{machkour2021terminating} and the \textit{model-X} knockoff method~\cite{candes2018panning}. Among these two frameworks, only the \textit{T-Rex} selector is scalable to millions of variables and, therefore, applicable for real world high-dimensional GWAS with millions of variables in a reasonable computation time (see Figure~1 in~\cite{machkour2021terminating}, and~\cite{machkour2023ScreenTRex}). Recently, the \textit{T-Rex+GVS} selector was proposed~\cite{machkour2022TRexGVS}. It is an extension of the \textit{T-Rex} selector that uses the elastic net~\cite{zou2005regularization} as a base selector to perform grouped variable selection and, thereby, significantly increases the TPR compared to the original \textit{T-Rex} selector in grouped variable selection problems. Unfortunately, however, this performance increase came at the cost of an increased computation time, which (in practice) reduces its scalability to very large GWAS.

\emph{Original Contributions}: We propose the \textit{informed elastic net (IEN)}, a new base selector that performs grouped variable selection while significantly reducing the computation time compared to the original \textit{elastic net (EN)}. We prove that the proposed \textit{IEN} 
\begin{enumerate}
\item can be formulated as a \textit{Lasso}-type optimization problem (Theorem~\ref{Theorem: Lasso-type optimization problem}) and, therefore, can be solved efficiently in a forward-selection manner, as required by the \textit{T-Rex} framework, using the \textit{LARS} algorithm~\cite{efron2004least} and
\item exhibits a grouping effect (Theorem~\ref{Theorem: IEN grouping effect}) that is similar to that of the elastic net.
\end{enumerate}
Additionally, we validate empirically that the proposed \textit{IEN}, as suggested by Theorems~\ref{Theorem: Lasso-type optimization problem} and~\ref{Theorem: IEN grouping effect},
\begin{enumerate}
\item produces solution paths that are similar to the ones of the \textit{EN}~(see Figure~\ref{fig: EN and IEN solution paths}),
\item significantly reduces the computation time compared to the \textit{EN} when incorporated into the \textit{T-Rex} framework as the base selector (see Figure~\ref{fig: relative_cpu_time_T_Rex_EN_IEN}), and
\item has the same TPR as the \textit{T-Rex+GVS} selector using the \textit{EN} while achieving a much lower FDR in a simulated GWAS (see Figure~\ref{fig: simulated GWAS IEN}).
\end{enumerate}

An open source implementation of the proposed method is available in the R package \textit{TRexSelector} on CRAN~\cite{machkour2022TRexSelector}.

Organization: Section~\ref{sec: The T-Rex+GVS Selector} briefly revisits the existing \textit{T-Rex} framework. In Section~\ref{sec: Proposed: The Informed Elastic Net (IEN)}, the proposed \textit{informed elastic net (IEN)} is presented. Section~\ref{sec: Numerical Experiments} compares the proposed \textit{informed elastic net} against the elastic net in terms of performance and computation time, while Section~\ref{sec: Simulated GWAS} presents the results of a simulated GWAS, in which the proposed \textit{IEN} is used as the base selector within the \textit{T-Rex} framework. Section~\ref{sec: Conclusion} concludes the paper and the proofs of Theorems~\ref{Theorem: Lasso-type optimization problem} and~\ref{Theorem: IEN grouping effect} are deferred to the appendix.

\section{The T-Rex+GVS Selector}
\label{sec: The T-Rex+GVS Selector}
This work builds upon the Terminating-Random Experiments (\textit{T-Rex}) selector~\cite{machkour2021terminating}, a recently developed variable selection framework for high-dimensional data that controls any user-defined target FDR (Theorem~1 in~\cite{machkour2021terminating}) while maximizing the number of selected variables. This is achieved by mathematically modeling and fusing the solutions of $K$ early terminated random experiments, where computer-generated dummy variables compete with real variables. Figure~\ref{fig: T-Rex selector framework} shows a schematic overview of the \textit{T-Rex} selector. The main steps are briefly revisited in the following. First, $K$ dummy matrices $\XK_{k} \in \mathbb{R}^{n \times L}$, $k = 1, \ldots, K$, are generated (each containing $L$ dummy predictor vectors). The elements of the dummy predictors can be generated from any univariate probability distribution with finite mean and variance, e.g., the standard normal distribution (Theorem~2 in~\cite{machkour2021terminating}). Second, the dummy matrices are appended to the original predictor matrix $\X = [\x_{1} \, \cdots \, \x_{p}] \in \mathbb{R}^{n \times p}$ to obtain the extended predictor matrices $\XWK_{k} = [\X \,\, \XK_{k}]$, $k = 1, \ldots, K$. Third, early terminated random experiments are conducted by feeding $\XWK_{k}$, $k = 1, \ldots, K$, and $\y$ into a base forward variable selection algorithm, such as the \textit{LARS} algorithm~\cite{efron2004least} or the \textit{Lasso}~\cite{tibshirani1996regression}, and terminating each selection process after $T$ dummies have been included. The obtained candidate variable sets $\C_{k, L}(T)$, $k = 1, \ldots, K$, are fused to obtain the relative occurrences of the original variables $\Phi_{T, L}(j)$, $j = 1, \ldots, p$. The final selected active set contains all variables whose relative occurrences exceed a voting threshold $v \in [0.5, 1)$, i.e., $\widehat{\A}_{L}(v^{*}, T^{*}) \coloneqq \big\lbrace j : \Phi_{T^{*}, L}(j) > v^{*}  \big\rbrace$. The extended calibration algorithm automatically chooses the values of $L$, $T^{*}$, and $v^{*}$, such that the FDR is controlled at the target level $\alpha \in [0, 1]$, while maximizing the number of selected variables.
%
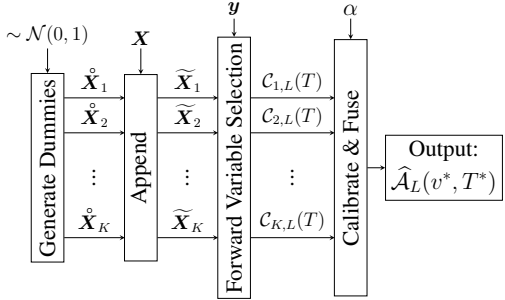
\begin{figure}[h]
\vspace{-20pt}
\begin{center}
\scalebox{0.62}{
\begin{tikzpicture}[>=stealth]

  \coordinate (orig)   at (0,0);
  \coordinate (sample)   at (1,0.5);
  \coordinate (merge)   at (3,0.5);
  \coordinate (varSelect)   at (5,0.5);
  \coordinate (tFDR)   at (15.5,-1.1);
  \coordinate (fuse)   at (7.5,0.5);
  \coordinate (output)   at (9.5,0.5);
  
  \coordinate (between_scale_rank)   at (0.5,0.31);
  \coordinate (X_prime_to_tFDR_point)   at (0.5,6);
  \coordinate (X_prime_to_tFDR_point_point)   at (9,6);
  \coordinate (X_prime_to_merge_point)   at (0.5,-2.5);
  \coordinate (center_to_tFDR_point)   at (5.00,3.7);
  \coordinate (tFDR_to_fuse_point)   at (7.5,3.7);
  \coordinate (tFDR_to_sample_point)   at (4,5);

  \coordinate (Arrow_N_GenDummy)   at (1,3.06);
  \coordinate (Arrow_X_indVar)   at (3,3.06);
  \coordinate (Arrow_targetFDR_tFDR)   at (10,5.9);
  
   \coordinate (inference_Arrow)   at (15,-1.06);
   \coordinate (fuse_Arrow)   at (16.1,2.06);
  
  \coordinate (vdots1)   at (2.0,0.4);
  \coordinate (vdots2)   at (4.0,0.4);
  \coordinate (vdots3)   at (6.25,0.4);
  
  \coordinate (fuse_node)   at (15.00,0.5);
  
  \node[draw, minimum width=.7cm, minimum height=4cm, anchor=center , align=center] (C) at (sample) {\rotatebox{90}{\Large Generate Dummies}};
  \node[draw, minimum width=.7cm, minimum height=4cm, anchor=center, align=center] (D) at (merge) {\rotatebox{90}{\Large Append}};   
  \node[draw, minimum width=.7cm, minimum height=5.5cm, anchor=center, align=center] (E) at (varSelect) {\rotatebox{90}{\Large Forward Variable Selection}};
  \node[draw, minimum width=.7cm, minimum height=5.5cm, anchor=center, align=center] (H) at (fuse) {\rotatebox{90}{\Large Calibrate \& Fuse}};
  \node[draw, minimum width=2.5cm, minimum height=.7cm, anchor=center, align=center] (N) at (output) {\Large Output: \\[0.3em] \Large $\widehat{\mathcal{A}}_{L}(v^{*}, T^{*})$};
  \node (J) at (vdots1) {\Large $\vdots$};
  \node (K) at (vdots2) {\Large $\vdots$};
  \node (L) at (vdots3) {\Large $\vdots$};
  
  \draw[->] (Arrow_N_GenDummy) -- node[above, pos = 0.1]{\large $\sim\mathcal{N}(0, 1)$} ($(C.90)$); 
  \draw[->] (Arrow_X_indVar) -- node[above, pos = 0.1]{\large $\X$} ($(D.90)$); 
     
  \draw[->] ($(C.0) + (0,1.5)$) -- node[above]{\large $\XK_{1}$} ($(D.0) + (-0.7,1.5)$);
  \draw[->] ($(C.0) + (0,0.75)$) -- node[above]{\large $\XK_{2}$} ($(D.0) + (-0.7,0.75)$);
  \draw[->] ($(C.0) + (0,-1.5)$) -- node[above]{\large $\XK_{K}$} ($(D.0) + (-0.7,-1.5)$);
     
  \draw[->] ($(D.0) + (0,1.5)$) -- node[above]{\large $\XWK_{1}$} ($(E.0) + (-0.7,1.5)$);
  \draw[->] ($(D.0) + (0,0.75)$) -- node[above]{\large $\XWK_{2}$} ($(E.0) + (-0.7,0.75)$);
  \draw[->] ($(D.0) + (0,-1.5)$) -- node[above]{\large $\XWK_{K}$} ($(E.0) + (-0.7,-1.5)$);
     
  \draw[->] ($(E.0) + (0,1.5)$) -- node[above]{\large $\C_{1, L}(T)$} ($(H.0) + (-0.7,1.5)$);
  \draw[->] ($(E.0) + (0,0.75)$) -- node[above]{\large $\C_{2, L}(T)$} ($(H.0) + (-0.7,0.75)$);
  \draw[->] ($(E.0) + (0,-1.5)$) -- node[above]{\large $\C_{K, L}(T)$} ($(H.0) +  (-0.7,-1.5)$);
     
  \draw[->] (center_to_tFDR_point) -- node[above, pos = 0.1]{\large $\y$} ($(E.90)$); 
  \draw[->] (tFDR_to_fuse_point) -- node[above, pos = 0.1]{\Large $\alpha$} ($(H.90)$);
 
  \coordinate (between_varSelect_fuse1)   at ($(E.0) + (2.75,1.5)$);
  \coordinate (between_varSelect_fuse2)   at ($(E.0) + (2.75,0.75)$);
  \coordinate (between_varSelect_fuse3)   at ($(E.0) + (2.75,-1.5)$);
 
  \draw[->] (H) -- (N);
\end{tikzpicture}}
\end{center}
\setlength{\abovecaptionskip}{2pt}
\setlength{\belowcaptionskip}{-3pt}
\caption{Sketch of the \textit{T-Rex} selector framework~\cite{machkour2021terminating,machkour2022TRexGVS}.}
\label{fig: T-Rex selector framework}
\end{figure}

%

The \textit{T-Rex+GVS} selector, where \textit{GVS} stands for grouped variable selection, was proposed recently as an extension of the \textit{T-Rex} selector~\cite{machkour2022TRexGVS}. It was proven that, when incorporating the elastic net as the base forward variable selector into the \textit{T-Rex} framework, the grouped variable selection property of the elastic net, whose solution is defined by
\begin{equation}
\hatbbeta_{\EN} \coloneqq \underset{\bbeta}{\arg\min} \, \| \y - \X\bbeta \|_{2}^{2} + \lambda_{1} \| \bbeta \|_{1} + \lambda_{2} \| \bbeta \|_{2}^{2},
\label{eq: elastic net}
\end{equation}
where $\lambda_{1}, \lambda_{2} > 0$, carries over to the \textit{T-Rex+GVS} selector~\cite{machkour2022TRexGVS}. Note that when incorporating the elastic net into the \textit{T-Rex} framework, the sparsity parameter $\lambda_{1}$ is automatically determined according to the user-defined target FDR level and does not require any tedious and computationally expensive tuning. The parameter $\lambda_{2}$ is chosen as suggested in~\cite{machkour2022TRexGVS} via a cross validation of Ridge regression solutions.

Moreover, for grouped variables, such as in GWAS, it was shown that using the \textit{elastic net} as the base selector within the \textit{T-Rex} framework leads to a significant increase in the TPR, while controlling the FDR at the user-defined target level.

\section{Proposed: The Informed Elastic Net (IEN)}
\label{sec: Proposed: The Informed Elastic Net (IEN)}
While the \textit{EN} achieves its grouping effect by penalizing $\| \bbeta \|_{2}^{2}$, this section presents a new \textit{GVS} method that incorporates the information of how the variables are grouped into its penalty term. We show that the proposed \textit{IEN} can be formulated as a \textit{Lasso}-type optimization problem (Theorem~\ref{Theorem: Lasso-type optimization problem}), so that it can be incorporated into the \textit{T-Rex} framework. We also analyze the grouping effect (Theorem~\ref{Theorem: IEN grouping effect}) and show that the \textit{IEN} boils down to the \textit{EN} when every variable is considered to be its own group (Corollary~\ref{Corollary: IEN boils down to EN}). 

In particular, the proposed \textit{IEN} uses single-linkage hierarchical clustering~\cite{murtagh2012algorithms} with the pairwise correlations of the original variables as a distance measure to cluster variables into groups of highly correlated variables, which are present in genomics data due to a phenomenon called linkage disequilibrium~\cite{reich2001linkage}. The obtained dendrogram from the hierarchical clustering can be cut at different levels to obtain $M$ disjoint groups of variables $\G_{1}, \ldots, \G_{M}$, where $M \leq p$. To represent the $m$th, $m = 1, \ldots, M$, group mathematically, we define the binary support vector $\1_{m} = [1_{m, 1} \, \cdots \, 1_{m, p} ]^{\top} \in \lbrace 0, 1 \rbrace^{p}$ that has one entries for variables in the $m$th group and zero entries otherwise. The corresponding group size is $p_{m} \coloneqq \sum_{j = 1}^{p} 1_{m, j}$. With these definitions in place, we define the proposed \textit{IEN}.
\begin{defn}[\textit{Informed elastic net (IEN)}]
Let $\lambda_{1}, \lambda_{2} > 0$ and let $p_{m}$, $m = 1, \ldots, M$, and $\1_{m} \in \lbrace 0, 1 \rbrace^{p}$ be the known group size and the binary support vector of the $m$th group, respectively. Then, the Lagrangian of the \textit{informed elastic net (IEN)} is defined by
\begin{equation}
\mathcal{L}_{\IEN}(\bbeta) \coloneqq \| \y  - \X\bbeta \|_{2}^{2} + \lambda_{1} \| \bbeta \|_{1} + \lambda_{2} \sum\limits_{m = 1}^{M}\dfrac{(\1_{m}^{\top} \bbeta)^{2}}{p_{m}}
\label{eq: Informed Elastic Net Lagrangian}
\end{equation}
and the solution of the \textit{IEN} is defined by
\begin{equation}
\hatbbeta \coloneqq \underset{\bbeta}{\arg\min} \, \mathcal{L}_{\IEN}(\bbeta).
\label{eq: solution of IEN}
\end{equation}
\end{defn}

The following theorem shows that the proposed \textit{IEN} can be cast as a \textit{Lasso}-type optimization problem and, therefore, can be solved efficiently and in a forward selection fashion, as required by the \textit{T-Rex} framework, using the \textit{LARS} algorithm:\footnote{Note that the proposed \textit{IEN} is fundamentally different from the \textit{group Lasso} approach in~\cite{yuan2006model}, since the solution path of the \textit{group Lasso} is not piecewise linear and, therefore, computationally much more expensive.}
\begin{thm}[\textit{Lasso}-type optimization problem]
Let $\X$, $\y$, and $\lambda_{1}, \lambda_{2} > 0$ be given and let $\boldsymbol{0}_{M}$ be the $M$-dimensional vector of zeros. Define
\vspace{-10pt}
\begin{equation}
\X^{\prime} \coloneqq \sqrt{\lambda_{2}} \cdot
\begin{bmatrix}
\X / \sqrt{\lambda_{2}}
\\
\1_{1}^{\top} / \sqrt{p_{1}}
\\
\vdots
\\
\1_{M}^{\top} / \sqrt{p_{M}}
\end{bmatrix},
\quad
\y^{\prime} \coloneqq 
\begin{bmatrix}
\y
\\
\boldsymbol{0}_{M}
\end{bmatrix}.
\label{eq: augmented X and augmented y for IEN}
\end{equation}
Then, the \textit{IEN} can be formulated as a \textit{Lasso}-type optimization problem, i.e.,
\vspace{-5pt}
\begin{equation}
\mathcal{L}_{\IEN}(\bbeta) = \| \y^{\prime} - \X^{\prime} \bbeta \|_{2}^{2} + \lambda_{1} \| \bbeta \|_{1}.
\label{eq: IEN as Lasso-type optimization problem}
\end{equation}
\label{Theorem: Lasso-type optimization problem}
\end{thm}
\begin{rmk}
Note that, in contrast to the \textit{EN}, the \textit{IEN} data augmentation presented in Theorem~\ref{Theorem: Lasso-type optimization problem} requires appending only $M$ additional rows to the original predictor matrix $\X$, while solving the elastic net in a forward selection manner requires appending $p$ rows to $\X$ (for details, see~\cite{zou2005regularization}). Since the number of variable groups $M$, especially in genomics data, is much smaller than the number of variables $p$ (i.e., $M \ll p$), the \textit{IEN} exhibits a significantly reduced computation time when $p$ is very large.
\end{rmk}
The next theorem shows that the proposed \textit{informed elastic net} exhibits a grouping effect, i.e., the difference of the averaged coefficients of any two variable groups is shrinked towards zero according to the maximum correlation between two variables from different groups:
\begin{thm}[\textit{IEN} grouping effect]
Define $\rho_{j, j^{\prime}} \coloneqq \x_{j}^{\top}\x_{j^{\prime}}$, where $\x_{j}$ and $\x_{j^{\prime}}$ are standardized predictors. Suppose that $\hat{\beta}_{j} \hat{\beta}_{j^{\prime}} > 0$ and, without loss of generality, $j \in \G_{1}$ and $j^{\prime} \in \G_{2}$. Then, it holds that
\begin{align}
&\dfrac{1}{\| \y \|_{2}} \Bigg| \dfrac{1}{p_{1}} \sum\limits_{g \in \G_{1}} \hat{\beta}_{g} - \dfrac{1}{p_{2}} \sum\limits_{g \in \G_{2}} \hat{\beta}_{g} \Bigg|
\\
&\,\, =
\dfrac{1}{\| \y \|_{2}} \bigg| \dfrac{\1_{1}^{\top}\hatbbeta}{p_{1}} - \dfrac{\1_{2}^{\top}\hatbbeta}{p_{2}} \bigg| \leq \dfrac{1}{\lambda_{2}} \sqrt{2 \bigg( 1 - \max_{j \in \G_{1}, j^{\prime} \in \G_{2}} \lbrace \rho_{j, j^{\prime}} \rbrace \bigg)}.
\label{eq: IEN grouping effect}
\end{align}
\label{Theorem: IEN grouping effect}
\end{thm}
\begin{cor}
The grouping effect of the proposed \textit{IEN} is identical to that of the \textit{EN}, when every variable is considered to be a group.
\label{Corollary: IEN boils down to EN}
\end{cor}
\begin{proof}
When every variable is considered to be a group, we have $M = p$, $p_{1} = \ldots = p_{p} = 1$, and the third summand in~\eqref{eq: Informed Elastic Net Lagrangian} boils down to $\lambda_{2} \sum_{m = 1}^{p} \beta_{m}^{2} = \lambda_{2} \| \bbeta \|_{2}^{2}$ and, thus,~\eqref{eq: elastic net} and~\eqref{eq: solution of IEN} are equivalent.
\end{proof}
\begin{rmk}
Note that, as desired, the difference of the averaged coefficients of any two variable groups are exactly zero if two variables from different groups are perfectly correlated.
\end{rmk}

\section{Numerical Experiments}
\label{sec: Numerical Experiments}
%
\begin{figure}[!htbp]
  \centering
  \subfloat[\textit{Elastic net (EN)} solution path.]{
  		\scalebox{1}{
  			\includegraphics[width=0.75\linewidth]{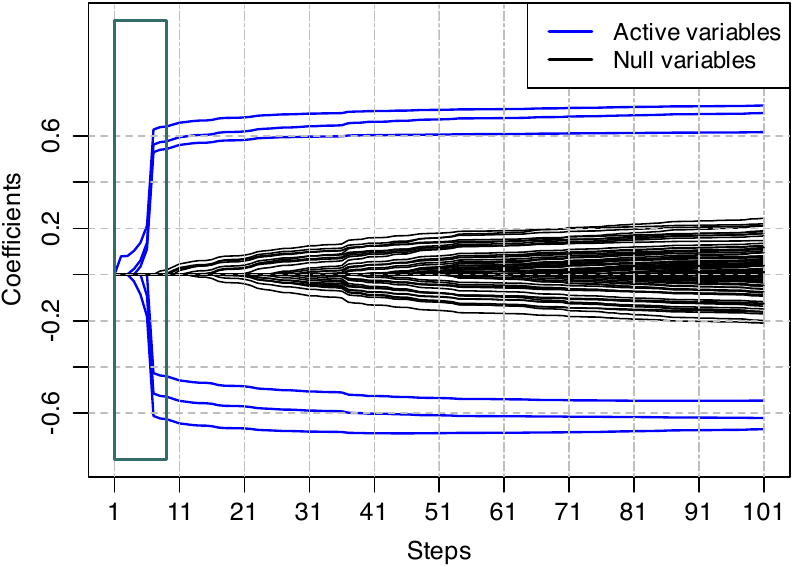}
  		}
   		\label{fig: EN_solution_path_p_100}
   }
	\vspace*{-0.5em}
\\
  \subfloat[Proposed: \textit{Informed elastic net (IEN)} solution path.]{
  		\scalebox{1}{
  			\includegraphics[width=0.75\linewidth]{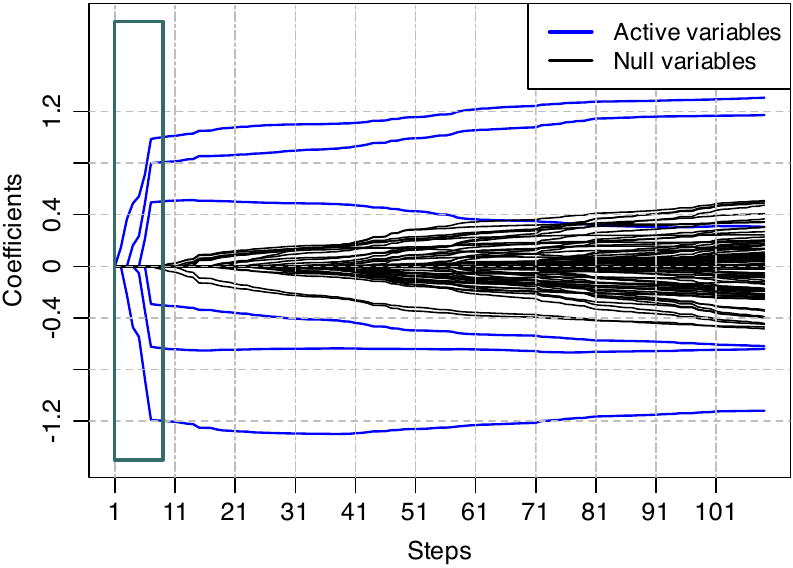}
  		}
   		\label{fig: IEN_solution_path_p_100}
   }
\setlength{\abovecaptionskip}{5pt}
\setlength{\belowcaptionskip}{-15pt}
  \caption{Solution paths of the (a) \textit{EN}~\cite{zou2005regularization} and (b) \textit{IEN}.}
  \label{fig: EN and IEN solution paths}
\end{figure}
%
In this section, we evaluate the grouping effect of the \textit{EN} and the proposed \textit{IEN} and their relative computation times when incorporated into the \textit{T-Rex} framework.
\subsection{Grouping Effect and Solution Path}
\label{subsec: Grouping Effect and Solution Path}
To compare the solution paths of the \textit{elastic net} and the proposed \textit{informed elastic net}, we first generate a setting with $p = 100$ and $6$ active standard normal variables that are split into two independent groups of highly correlated variables $\G_{1}$ and $\G_2$ (i.e., any pair of the three variables in one group has a correlation of $0.75$). The remaining $94$ null variables are sampled independently from the standard normal distribution. The coefficient vector $\bbeta = [ \beta_{1} \,\, \cdots \,\, \beta_{p} ]^{\top}$ of the variables is  chosen as follows: $\beta_{j} = 1$ for $j \in \G_{1}$, $\beta_{j} = -1$ for $j \in \G_{2}$ and $\beta_{j} = 0$ otherwise. The response variable $\y$ is generated from the linear regression model $\y = \X\bbeta + \bepsilon$, where $\bepsilon \sim \mathcal{N}(\boldsymbol{0}, \sigma^{2}\boldsymbol{I})$ is the noise vector. In order to make the grouping effect and the distinction between nulls and active variables visually noticeable, we generated $n = 150$ samples and set the noise variance $\sigma^{2}$ such that the signal-to-noise-ratio $\SNR \coloneq \Var(\X\bbeta) / \Var(\bepsilon) = 3$. To obtain the binary support vectors that are required for the data augmentation presented in Theorem~\ref{Theorem: Lasso-type optimization problem}, we use single-linkage hierarchical clustering with the pairwise correlations between variables as distance measures and cut the resulting dendrogram at the maximum height that satisfies the conservative condition that there exist no two variables from different clusters with a correlation higher than $0.2$.

In Figure~\ref{fig: EN and IEN solution paths}, it can be observed that the \textit{EN} and the \textit{IEN} both exhibit the grouping effect in the sense that the coefficients of the two groups of highly correlated active variables are increased in a correlated fashion. We also observe that the grouping effect of the proposed \textit{IEN} is slightly weaker than that of the \textit{EN}. However, since we use the \textit{IEN} as the base selector within the \textit{T-Rex} framework, which terminates the solution paths of all random experiments early, we are primarily interested in the early steps, where a sufficient grouping is observed for both methods, as illustrated in the boxed regions of Figure~\ref{fig: EN and IEN solution paths}.

\subsection{Relative Computation Time}
\label{subsec: Relative Computation Time}
We compare the relative computation times of one random experiment of 
\begin{enumerate}
\item the original \textit{T-Rex} selector with the \textit{LARS} algorithm as the base selector,
\item the \textit{T-Rex+GVS} selector with the \textit{EN} base selector, and
\item the \textit{T-Rex+GVS} selector with the \textit{IEN} as the base selector
\end{enumerate}
in a setting that is as described in Section~\ref{subsec: Grouping Effect and Solution Path}, except that we fix the number of samples to $n = 50$, set the correlation cutoff of the dendrogram to $0.5$ and increase $p$ from $100$ to $5{,}000$. The computation times are averaged over $50$ Monte Carlo replications. In Figure~\ref{fig: relative_cpu_time_T_Rex_EN_IEN}, we see that with growing number of variables $p$, the relative computation time of the \textit{T-Rex+GVS} selector with the proposed \textit{IEN} as the base selector decreases significantly. In particular, the savings in computation time start to manifest in larger settings with $p \geq 500$ variables, where the \textit{EN} always needs to augment $\X$ with $p$ (i.e., number of variables) rows, while the \textit{IEN} only requires augmenting $\X$ with $M$ (i.e., number of groups) rows, where $M \ll p$.
%
\begin{figure}[!htbp]
  \centering
  		\scalebox{1}{
  			\includegraphics[width=0.6\linewidth]{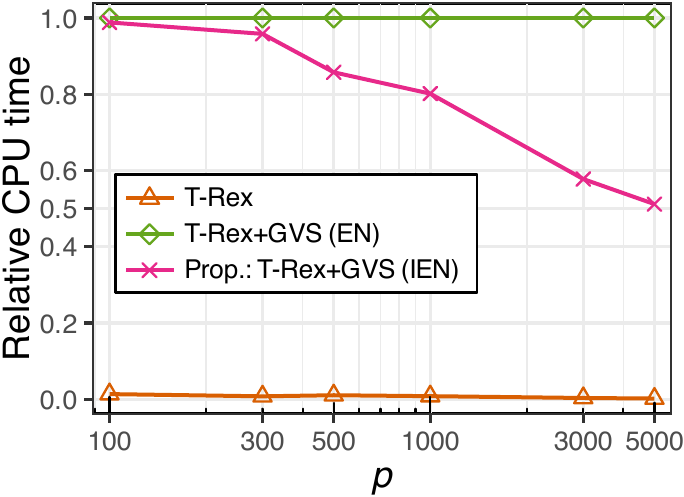}
  		}
\setlength{\abovecaptionskip}{2pt}
\setlength{\belowcaptionskip}{-10pt}
  \caption{Relative computation times of one random experiment with $L = p$ and $T = 1$ of the \textit{T-Rex}, \textit{T-Rex+GVS (EN)}, and the proposed \textit{T-Rex+GVS (IEN)}.}
  \label{fig: relative_cpu_time_T_Rex_EN_IEN}
\end{figure}
%
\begin{figure}[h]
  \centering
  \hspace{-0.25cm}
  \subfloat{
  		\scalebox{1}{
  			\includegraphics[width=0.47\linewidth]{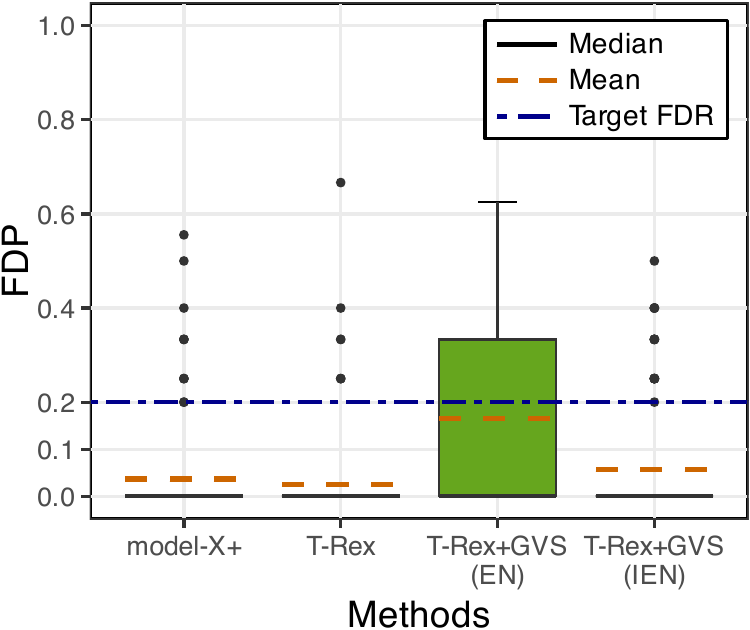}
  		}
   		\label{fig: FDP_simulated_GWAS_IEN}
   }
   \hspace{-0.45cm}
  \subfloat{
  		\scalebox{1}{
  			\includegraphics[width=0.47\linewidth]{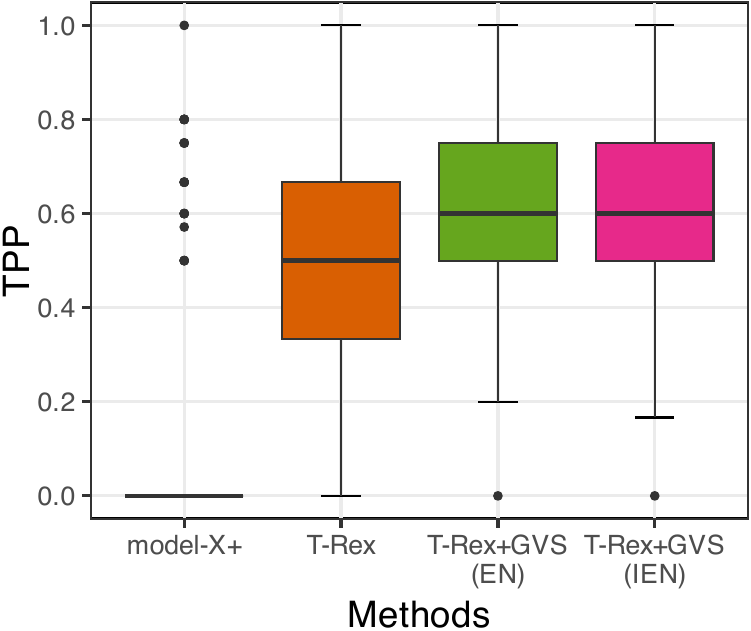}
  		}
   		\label{fig: TPP_simulated_GWAS_IEN}
   }
	\setlength{\abovecaptionskip}{0pt}
	\setlength{\belowcaptionskip}{-10pt}
  \caption{FDP and TPP performances in the simulated GWAS.}
  \label{fig: simulated GWAS IEN}
\end{figure}
%

\section{Simulated GWAS}
\label{sec: Simulated GWAS}
This section presents the results of a GWAS simulation study for which we generate $100$ data sets that each contain $500$ cases and $200$ controls using the software HAPGEN2~\cite{su2011hapgen2}. That is, we input real world haplotypes from the HapMap $3$ project~\cite{international2010integrating} into HAPGEN2 and the software generates $100$ predictor matrices that each contain $1{,}000$ SNPs as columns. The phenotype vector $\y$ contains ones for cases and zeros for controls. For each data set, we generate $10$ true active SNPs, while the remaining ones are nulls. The standard preprocessing of genomics data and the evaluation of the results is carried out as described in~\cite{machkour2022TRexGVS}. As suggested in~\cite{machkour2021terminating}, we run $K = 20$ random experiments for the~\textit{T-Rex} methods.

Figure~\ref{fig: simulated GWAS IEN} presents the box plots of the false discovery proportion (FDP) and the true positive proportion (TPP) and the means of the FDP, which are estimates of the FDR, since the FDR is defined as the expectation of the FDP. First, we observe that all methods control the FDR. The median TPP of the benchmark \textit{model-X} knockoff method is zero. The existing \textit{T-Rex+GVS (EN)} selector shows a significant improvement in TPP compared to the original \textit{T-Rex} selector. As desired, the significant increase in TPP is also achieved by the proposed \textit{T-Rex+GVS (IEN)} selector, while the FDP is much lower compared to that of the \textit{T-Rex+GVS (EN)} selector. Thus, in this GWAS use-case, the proposed \textit{T-Rex+GVS (IEN)} selector dominates the existing \textit{T-Rex+GVS (EN)} selector, while exhibiting a much lower computation time, especially in large-scale high-dimensional settings (see Figure~\ref{fig: relative_cpu_time_T_Rex_EN_IEN}).

\section{Conclusion}
\label{sec: Conclusion}
We have proposed the \textit{informed elastic net (IEN)}, a fast grouped variable selection method for high-dimensional settings. We incorporated it into the \textit{T-Rex} framework and showed that it has a significantly reduced computation time in large-scale high-dimensional settings and a better performance in a simulated GWAS compared to the \textit{T-Rex+GVS (EN)}~\cite{machkour2022TRexGVS}.

\appendix
\begin{proof}[Proof of Theorem~\ref{Theorem: Lasso-type optimization problem}]
First, we rewrite~\eqref{eq: IEN as Lasso-type optimization problem} as follows:
\begin{equation}
\mathcal{L}_{\IEN} = \big( \y^{\prime\top} \y^{\prime} - 2\bbeta^{\top}\X^{\prime\top}\y^{\prime} + \bbeta^{\top}\X^{\prime\top}\X^{\prime}\bbeta \big) + \lambda_{1} \| \bbeta \|_{1}.
\label{eq: proof - Theorem - Lasso-type optimization problem 1}
\end{equation}
Second, note that $\X^{\prime\top}\X^{\prime} = \X^{\top} \X + \lambda_{2} \sum_{m = 1}^{M} \frac{\1_{m} \1_{m}^{\top}}{p_{m}}$. Then, plugging~\eqref{eq: augmented X and augmented y for IEN} into~\eqref{eq: proof - Theorem - Lasso-type optimization problem 1} yields
\begin{align}
\mathcal{L}_{\IEN} &= \bigg[ \y^{\top}\y - 2\bbeta^{\top}\X^{\top}\y
\\[-1em]
&\qquad + \bbeta^{\top} \bigg( \X^{\top} \X + \lambda_{2} \sum\limits_{m = 1}^{M} \dfrac{\1_{m} \1_{m}^{\top}}{p_{m}} \bigg) \bbeta \bigg] + \lambda_{1} \| \bbeta \|_{1}
\\
&= \y^{\top}\y - 2\bbeta^{\top}\X^{\top}\y + \bbeta^{\top}\X^{\top}\X\bbeta
\\
&\qquad + \lambda_{2} \sum\limits_{m = 1}^{M} \dfrac{\bbeta^{\top}\1_{m} \1_{m}^{\top}\bbeta}{p_{m}} + \lambda_{1} \| \bbeta \|_{1}
\\[-0.5em]
&= \| \y - \X\bbeta \|_{2}^{2} + \lambda_{1} \| \bbeta \|_{1} + \lambda_{2} \sum\limits_{m = 1}^{M}\dfrac{(\1_{m}^{\top} \bbeta)^{2}}{p_{m}}. \tag*{\qedhere}
\end{align}
\label{proof: Theorem - Lasso-type optimization problem}
\end{proof}
\vspace{-1.5em}
\begin{proof}[Proof of Theorem~\ref{Theorem: IEN grouping effect}]
Define $\boldsymbol{\hat{r}} \coloneqq \y - \X\hatbbeta$. Taking the first derivative of~\eqref{eq: Informed Elastic Net Lagrangian} and setting it equal to zero, we obtain
\begin{align}
&\dfrac{\partial \mathcal{L}_{\IEN}(\bbeta)}{\partial \bbeta}\bigg|_{\bbeta = \hatbbeta}
\\[-0.75em]
&\,=
-2\X^{\top}\boldsymbol{\hat{r}} + \lambda_{1} \dfrac{\partial \| \bbeta \|_{1}}{\partial \bbeta}\bigg|_{\bbeta = \hatbbeta} + 2\lambda_{2} \sum\limits_{m = 1}^{M}\dfrac{\1_{m} \1_{m}^{\top} \hatbbeta}{p_{m}} \overset{!}{=} \boldsymbol{0}.
\label{eq: proof - Theorem - IEN grouping effect 1}
\end{align}
The $j$th and $j^{\prime}$th equation of the system of equations in~\eqref{eq: proof - Theorem - IEN grouping effect 1}~are:
\begin{align}
& -2\x_{j}^{\top}\boldsymbol{\hat{r}} + \lambda_{1}\sign(\hat{\beta}_{j}) + 2\lambda_{2} \dfrac{\1_{1}^{\top} \hatbbeta}{p_{1}} = 0,
\label{eq: proof - Theorem - IEN grouping effect 2}
\\
& -2\x_{j^{\prime}}^{\top}\boldsymbol{\hat{r}} + \lambda_{1}\sign(\hat{\beta}_{j^{\prime}}) + 2\lambda_{2} \dfrac{\1_{2}^{\top} \hatbbeta}{p_{2}} = 0.
\label{eq: proof - Theorem - IEN grouping effect 3}
\end{align}
Subtracting~\eqref{eq: proof - Theorem - IEN grouping effect 3} from~\eqref{eq: proof - Theorem - IEN grouping effect 2} and noting that $\sign(\hat{\beta}_{j}) = \sign(\hat{\beta}_{j^{\prime}})$:
\begin{equation}
2 ( \x_{j^{\prime}}^{\top} - \x_{j}^{\top} ) \boldsymbol{\hat{r}} + 2\lambda_{2} \bigg( \dfrac{\1_{1}^{\top} \hatbbeta}{p_{1}} - \dfrac{\1_{2}^{\top} \hatbbeta}{p_{2}} \bigg) = 0.
\label{eq: proof - Theorem - IEN grouping effect 4}
\end{equation}
Using~\eqref{eq: proof - Theorem - IEN grouping effect 4}, we obtain
\begin{align}
&\dfrac{1}{\| \y \|_{2}} \bigg| \dfrac{\1_{1}^{\top}\hatbbeta}{p_{1}} - \dfrac{\1_{2}^{\top}\hatbbeta}{p_{2}} \bigg|
= \dfrac{1}{\lambda_{2} \| \y \|_{2}} \big| ( \x_{j}^{\top} - \x_{j^{\prime}}^{\top} ) \boldsymbol{\hat{r}} \big|
\\
&\quad \leq \dfrac{1}{\lambda_{2}} \| \x_{j} - \x_{j^{\prime}} \|_{2} \dfrac{\| \boldsymbol{\hat{r}} \|_{2}}{\| \y \|_{2}}
= \dfrac{1}{\lambda_{2}} \sqrt{2(1 - \rho_{j, j^{\prime}})} \dfrac{\| \boldsymbol{\hat{r}} \|_{2}}{\| \y \|_{2}}
\\
&\quad \leq \dfrac{1}{\lambda_{2}} \sqrt{2(1 - \rho_{j, j^{\prime}})},
\label{eq: proof - Theorem - IEN grouping effect 5}
\end{align}
where the inequality and the equation in the second line follow from the Cauchy-Schwarz inequality and from $\| \x_{j} - \x_{j^{\prime}} \|_{2} = \sqrt{ \| \x_{j} \|_{2}^{2} + \| \x_{j^{\prime}} \|_{2}^{2} - 2\x_{j}^{\top}\x_{j^{\prime}}} = \sqrt{1 + 1 - 2\rho_{j, j^{\prime}}}$, respectively. The last line follows from the fact that $\hatbbeta$ is the minimizer of~\eqref{eq: Informed Elastic Net Lagrangian} and, therefore, $\mathcal{L}(\hatbbeta) \leq \mathcal{L}(\boldsymbol{0}) \Leftrightarrow \| \boldsymbol{\hat{r}} \|_{2}^{2} + \lambda_{1} \| \hatbbeta \|_{1} + \lambda_{2} \sum_{m = 1}^{M} \frac{(\1_{m}^{\top} \hatbbeta)^{2}}{p_{m}} \leq \| \y \|_{2}^{2} \Rightarrow \| \boldsymbol{\hat{r}} \|_{2} \leq \| \y \|_{2}$. Since the inequality in~\eqref{eq: proof - Theorem - IEN grouping effect 5} holds for all $j \in \G_{1}$ and $j^{\prime} \in \G_{2}$, the smallest upper bound is given by the largest $\rho_{j, j^{\prime}}$, i.e.,
\begin{equation}\belowdisplayskip=-12pt
\dfrac{1}{\| \y \|_{2}} \bigg| \dfrac{\1_{1}^{\top}\hatbbeta}{p_{1}} - \dfrac{\1_{2}^{\top}\hatbbeta}{p_{2}} \bigg| \leq \dfrac{1}{\lambda_{2}} \sqrt{2 \bigg( 1 - \max_{j \in \G_{1}, j^{\prime} \in \G_{2}} \lbrace \rho_{j, j^{\prime}} \rbrace \bigg)}.
\label{eq: proof - Theorem - IEN grouping effect 6}
\end{equation}
\label{proof: Theorem - IEN grouping effect}
\end{proof}
\clearpage
\bibliographystyle{IEEEtran}
\bibliography{bibliography}

\typeout{get arXiv to do 4 passes: Label(s) may have changed. Rerun}

\end{document}